\documentclass[aps,prb,twocolumn,preprintnumbers,amsmath,amssymb,floatfix]{revtex4-1}
\usepackage{graphicx,color,xcolor,mathtools,multirow}

\newcommand\bigforall{\mbox{\huge $\forall$}} 
\newenvironment{proof}{\noindent\textbf{Proof.\ }}{\hspace*{\fill}$\blacksquare$\medskip}
\newtheorem{theorem}{Theorem}
\newcommand{\tmop}[1]{\ensuremath{\operatorname{#1}}}
\newcommand{\tmtexttt}[1]{{\ttfamily{#1}}}

\begin{document}

\title{Charge-density-wave phases of the generalized $t$-$V$ model}

\author{M.\ Szyniszewski}

\email[Contact e-mail: ]{mszynisz@gmail.com}

\affiliation{Physics Department, Lancaster University, Lancaster LA1 4YB, UK}

\affiliation{NoWNano DTC, University of Manchester, Manchester M13 9PL, UK}

\date{\today}

\begin{abstract}
  The one-dimensional extended $t$-$V$ model of fermions on a lattice is a model with
	repulsive interactions of finite range that exhibits a transition between a Luttinger liquid
  conducting phase and a~Mott insulating phase. It is known 
	that by tailoring the potential energy of the insulating system, one can force a phase transition
	into another insulating phase. We show how to construct all possible 
	charge-density-wave phases of the system at low critical densities in the atomic
  limit. Higher critical densities are investigated by a brute-force analysis
  of the possible finite unit cells of the Fock states.
\end{abstract}

\maketitle

\section{Introduction}

Low-dimensional quantum systems exhibit unique electronic properties impossible to realize in bulk. A proper theoretical understanding of these materials could allow us to use them in a variety of novel electronic devices. Many low-dimensional systems show interesting behavior, such as the presence of phases that cannot be explained using classical theory.{\cite{Giamarchi2003}} One example of such phases is a Mott insulator, which has applications ranging from high-temperature superconductors {\cite{Kohsaka2008}} to a new type of energy-efficient field effect transistor with fast switching times.{\cite{Newns1998}} The research into the subject of one- and two-dimensional Mott transistors is currently ongoing.{\cite{Inoue2008,Son2011,Nakano2012}} However, to make an efficient Mott insulating device we first need an accurate description of the possible phase diagram of the system. 

All one-dimensional quantum liquids fall into the universality class of a Luttinger liquid,{\cite{Giamarchi2003, Haldane1981, Cazalilla2003, Cazalilla2011}} which is a concept similar to the Fermi liquid in higher dimensions. The theory of Luttinger liquids has already proven to be applicable in experiments dealing with carbon nanotubes,{\cite{Bockrath1998}} the fractional quantum Hall effect,{\cite{Chang1996, Chang2003}} and crystals of trapped ions.{\cite{Schneider2012}}

The extended $t$-$V$ model{\cite{GomezSantos1993}} is a perfect subject for 
investigation, exhibiting both a conducting Luttinger liquid phase and a Mott 
insulator phase. Recent studies {\cite{Schmitteckert2004, Mishra2011,
Dalmonte2015, Venuti2015}} have shown that one can achieve multiple insulating quantum 
phases in this system and thus there is a question of how to provide a 
systematic description of these insulators when the interaction range is 
varied.

In this article we investigate charge-density-wave (CDW) insulating phases 
that are present in a system with specific potential energy. 
The outline of this work is as follows: in Sec.\ \ref{model} we present 
a generalization of the $t$-$V$ model and describe its known behavior; then in Sec.\ \ref{analytic} we 
investigate CDW phases at low critical densities, while higher critical densities are 
analyzed in Sec.\ \ref{numeric} using brute-force sampling of the partial basis; finally, in Section\ \ref{summary} we
conclude and discuss possible future work.

\section{The generalized $t$-$V$ model} \label{model}

The Hamiltonian of the generalized $t$-$V$ model of spinless fermions in
a one-dimensional periodic chain of length $L$ is given by {\cite{GomezSantos1993,Szyniszewski2015}}
\begin{equation}
  \hat{\mathcal{H}} = - t \sum_{i = 1}^L \left( \hat{c}^{\dag}_i \hat{c}_{i +
  1} + \text{h.c.} \right) + \sum_{i = 1}^L \sum_{m = 1}^p U_m \hat{n}_i
  \hat{n}_{i + m},
\end{equation}
where $\hat{c}_i$ and $\hat{c}^{\dag}_i$ are the fermionic annihilation and
creation operators on site $i$, $\hat{n}_i = \hat{c}_i^{\dag} \hat{c}_i$ is
the particle number operator, $p$ is the maximum range of the interaction, $t$ is
the strength of the kinetic energy, and $U_m\geq 0$ is the potential between two fermions
that are $m$ sites apart from each other (for $m>p$, $U_m=0$). Notice that the 
potential part is diagonal in the basis consisting of Fock states.

The hopping part is assumed to have much smaller strength than the potential part, i.e., $\forall_{m\le p}
\ t \ll U_m$, which means that two fermions are likely to be more than $p$
sites away from each other. This also leads to a~critical Mott insulating
density $Q = \frac{1}{p + 1}$, where there is a huge energy penalty to moving
any one fermion. By shaking a system in the critical density, one
can create a CDW. G{\'o}mez-Santos {\cite{GomezSantos1993}}
introduces another very important assumption:
\begin{equation}
  \underset{m}{\bigforall}\ U_m < \frac{U_{m + 1} + U_{m - 1}}{2}.
  \label{condU}
\end{equation}
If the fermion-fermion distance is required to be less than $p$ sites (due
to high density in the system), then the particles will want to be as spread out as
possible. One can for example consider two similar systems,
both in Fock states,
which are different only by fermion chains: $(\cdots\nobreak\bullet\nobreak\circ\nobreak\circ\nobreak
\bullet\nobreak\circ\nobreak\circ\nobreak\bullet\nobreak\cdots)$ and $(\cdots\nobreak
\bullet\nobreak\circ\nobreak\bullet\nobreak\circ\nobreak\circ\nobreak\circ
\nobreak\bullet\nobreak\cdots)$, where
$\bullet$ and $\circ$ denote occupied and empty sites respectively. Assumption
(\ref{condU}) tells us that the first system will always have lower energy 
regardless of the maximum range of interactions. 
Thus, there are critical densities $Q = \frac{1}{m}$, when $m = 1, 2,
\ldots, p + 1$, at which the system is insulating and has the following unperturbed $(t
\rightarrow 0)$ ground state:
\begin{equation}
  \bullet \underbracket[.4pt]{\circ \circ \cdots \circ}_{
	\mbox{\scriptsize
	$\begin{array}{c}
    1 / Q - 1 \\
		{\tmop{sites}}
  \end{array}$}
	} \bullet
  \underbracket[.4pt]{\circ \circ \cdots \circ}_{1 / Q - 1} \bullet
  \underbracket[.4pt]{\circ \circ \cdots \circ}_{1 / Q - 1} \cdots.
\end{equation}
and the energy density is $E Q / N = Q U_{1 / Q}$ if $m > 
\frac{p + 1}{2}$, where $N = \sum_i \hat{n}_i$ is the total number of particles in the system. By converting condition (\ref{condU}) into
\begin{equation}
  \underset{m}{\bigforall} \ \frac{U_{m + 1} + U_{m - 1} - 2 U_m}{a^2} > 0,
\end{equation}
where $a$ is the lattice constant, we can immediately see that this assumption
is a discrete version of the (continuous) inequality
{\vspace{3pt plus 1pt minus 1pt}
\begin{equation}
  U'' (r) > 0. \label{condUcont}
\end{equation}
}
One can easily check that assumption (\ref{condUcont}) holds for Coulomb and dipole
potentials. However, in principle, a potential that does not satisfy such a condition could
also be considered (such as the P{\"o}schl-Teller potential used in the description of ultracold atomic gases).

Models with interaction range $p = 2$ and $Q = 1 / 3$ and $1 / 2$ have already been
analyzed in Refs.\ {\onlinecite{Schmitteckert2004}} and {\onlinecite{Mishra2011}}. 
Depending on the strength of the potentials, one can have
different phases in the system: there can be multiple CDW insulating phases, 
a long-range bond-order phase, and even a Luttinger
liquid phase, despite the existence of a critical (``insulating'') density.

The objective of this work is to generalize this result for all interaction
ranges. However, to simplify the problem, we shall assume the atomic limit ($t
= 0$), in which the only phases that will be encountered are CDW insulators.

\section{Low critical densities in the atomic limit} \label{analytic}

\subsection{Critical density $Q = 1 / (p + 1)$}

In the trivial case of $Q = 1/(p+1)$, the ground-state energy is always equal to zero. The
ground-state configuration is
\begin{equation}
  \bullet \underbracket[.4pt]{\circ \circ \cdots \circ}_p \bullet
  \underbracket[.4pt]{\circ \circ \cdots \circ}_p \bullet \underbracket[.4pt]{\circ
  \circ \cdots \circ}_p \cdots.
\end{equation}
Such a ground state is $(p + 1)$-fold degenerate since the energy is invariant under translation.

\subsection{Critical density $Q = 1 / p$}

Let us now show how to construct the CDW phase for a system with any $p$ and
with critical density $Q = \frac{1}{p}$. Firstly, assume that $U_p$ 
is low enough to ensure that the preferable distance between two fermions 
is always $p$ and thus we can say that $U_p$ orders the fermions in the
ground state; for example a chain $\bullet\nobreak\underbracket[.4pt]{\circ \circ \cdots
\circ}_{p - 1}\nobreak\bullet\nobreak\underbracket[.4pt]{\circ \circ \cdots \circ}_{p -
1} \bullet$ has lower energy than $\bullet\nobreak\underbracket[.4pt]{\circ \circ \circ
\cdots \circ}_p\nobreak\bullet\nobreak\underbracket[.4pt]{\circ \cdots \circ}_{p -
2} \bullet$. The ground state must have the simple form
\begin{equation}
  \bullet \underbracket[.4pt]{\circ \circ \cdots \circ}_{p - 1} \bullet
  \underbracket[.4pt]{\circ \circ \cdots \circ}_{p - 1} \bullet
  \underbracket[.4pt]{\circ \circ \cdots \circ}_{p - 1} \cdots
\end{equation}
and its energy is $E_1 = \frac{L}{p} U_p = N U_p$.

Now, let us assume that $U_{p - 1}$ is low enough to order the fermions. We
could use a series of $\bullet\nobreak\underbracket[.4pt]{\circ \circ \cdots
\circ}_{p - 2}$ sections, but then we would not achieve the correct density
$1 / p$. However, by addition of sections $\bullet\nobreak\underbracket[.4pt]{\circ
\circ \cdots \circ}_p$ we can tailor the density without changing
the energy of the system. Thus, the ground-state configuration is
\begin{equation}
	\mbox{\footnotesize
  $\begin{array}{|l|}
    \hline
    \bullet \underbracket[.4pt]{\circ \circ \cdots \circ}_{p - 2} \bullet
    \underbracket[.4pt]{\circ \circ \cdots \circ}_p\\
    \hline
  \end{array} \begin{array}{|l|}
    \hline
    \bullet \underbracket[.4pt]{\circ \circ \cdots \circ}_{p - 2} \bullet
    \underbracket[.4pt]{\circ \circ \cdots \circ}_p\\
    \hline
  \end{array} \begin{array}{|l|}
    \hline
    \bullet \underbracket[.4pt]{\circ \circ \cdots \circ}_{p - 2} \bullet
    \underbracket[.4pt]{\circ \circ \cdots \circ}_p\\
    \hline
  \end{array} \cdots$
	},
\end{equation}
which gives us the correct density $Q = \frac{2}{2 p} = \frac{1}{p}$ and
energy $E_2 = \frac{L}{2 p} U_{p - 1} = \frac{N}{2} U_{p - 1}$. The boxes are
present to show that we have correctly counted the energy and particle density. In general,
however, the whole subspace of the unperturbed ground states would include Fock states 
in which sections with $p - 2$ holes could be beside each other, unless they would change 
the energy of the system.

If one follows this prescription, in the $n$-th step the following ground
state is obtained:
\begin{widetext}
\begin{equation}
  \begin{array}{|l|}
    \hline
    \bullet \underbracket[.4pt]{\circ \circ \cdots \circ}_{p - n}
    \overbracket[.4pt]{\bullet \underbracket[.4pt]{\circ \circ \cdots \circ}_p
    \bullet \underbracket[.4pt]{\circ \circ \cdots \circ}_p \bullet
    \underbracket[.4pt]{\circ \circ \cdots \circ}_p \cdots}^{n - 1
    \tmop{times}}\\
    \hline
  \end{array} \begin{array}{|l|}
    \hline
    \bullet \underbracket[.4pt]{\circ \circ \cdots \circ}_{p - n}
    \overbracket[.4pt]{\bullet \underbracket[.4pt]{\circ \circ \cdots \circ}_p
    \bullet \underbracket[.4pt]{\circ \circ \cdots \circ}_p \bullet
    \underbracket[.4pt]{\circ \circ \cdots \circ}_p \cdots}^{n - 1
    \tmop{times}}\\
    \hline
  \end{array} \cdots,
\end{equation}
\end{widetext}
with the energy $E_n = \frac{L}{1 + p - n + (n - 1) (p + 1)} U_{p + 1 - n} =
\frac{L}{n p} U_{p - n + 1} = \frac{N}{n} U_{p - n + 1}$. We can now calculate 
the exact conditions in which an arbitrary phase (designated by step $n$) will be dominant in the system:
\begin{equation}
\begin{split}
	\underset{k \neq n}{\bigforall}\ &E_n < E_k \\
	\Rightarrow \quad \underset{k \neq n}{\bigforall}\ &U_{p - n + 1} < \frac{n}{k} U_{p - k + 1}.
	\end{split}
\end{equation}
Renaming $\alpha = p - n + 1$ and $\beta = p - k + 1$,
\begin{equation}
  \underset{\beta \neq \alpha}{\bigforall}\ U_{\alpha} < \frac{p - \alpha +
  1}{p - \beta + 1} U_{\beta}.
\end{equation}
If this condition is fulfilled, then the phase with energy $E_{(\alpha)} =
\frac{N}{p - \alpha + 1} U_{\alpha}$ is dominant and the ground state
consists of $\frac{N}{p - \alpha + 1}$ blocks of $\bullet\nobreak\underbracket[.4pt]{\circ
\circ \cdots \circ}_{\alpha - 1}$ and $N \frac{p - \alpha}{p -
\alpha + 1}$ blocks of $\bullet\nobreak\underbracket[.4pt]{\circ \circ \cdots
\circ}_p$ and is $f$-fold degenerate, where:
\begin{equation}
  f = {\left\{ \begin{array}{ll}
    \left( \begin{array}{c}
      N\\
      N / (p - \alpha + 1)
    \end{array} \right) \cdot p & \tmop{if}\ 2 \alpha > p\\
    \left( \begin{array}{c}
      N \frac{p - \alpha}{p - \alpha + 1}\\
      N / (p - \alpha + 1)
    \end{array} \right) \cdot \frac{p (p - \alpha + 1)}{p - \alpha} &
    \tmop{otherwise}.
  \end{array} \right.}
\end{equation}
For $2 \alpha \leqslant p$ the problem with assessing the degeneracy is that
we need to exclude cases where blocks of $\bullet\nobreak\underbracket[.4pt]{\circ \circ
\cdots \circ}_{\alpha - 1}$ are too close to each other and thus
would increase the energy by $U_{2 \alpha}$.

\section{Higher critical densities in the atomic limit} \label{numeric}

For Mott insulators with $Q = \frac{1}{m}$, where $m = 1, \ldots, p - 1$, the number
of phases and their energies were found to be more difficult to obtain.
Rather than constructing the phases as done in Section \ref{analytic}, we shall
use a brute-force analysis of the basis for systems of finite size. 
Nevertheless, because we are interested in the thermodynamic limit, 
a periodic system of $L$ sites can be 
thought of as an infinite system with a~unit cell of $L$ sites.

\subsection{Properties of the system}

Sampling the full basis in systems with $Q > 1 / p$ is problematic, because
the dimension of the basis grows rapidly with a system size. However,
many of the Fock states will have the same energy. In particular, if two states
are cyclic permutations of each other, or cyclic permutations with inversion,
then such states must have the same energy due to the periodicity of the system.
Checking the full basis for cyclic permutations would still be computationally
quite a difficult task: firstly, because generating the full basis would take a
lot of memory, and secondly, because comparing all the states to check if they
are cyclic permutations would require a large computational time [of 
$O(2^{2 L})$]. An alternative approach to this problem is to consider 
the spaces between the fermions in our chain and to develop rules to generate 
a set of Fock states that will always contain ground states of the system.

\begin{theorem}
  \label{property1}For any basis state, the largest space between consecutive
  fermions must not be less than $1 / Q - 1$ sites.
\end{theorem}

\begin{proof}
  All spaces between consecutive fermions are equal only if all particles are
  $1 / Q - 1$ sites apart, i.e., the configuration is
  \begin{equation}
    \bullet \underbracket[.4pt]{\circ \circ \cdots \circ}_{1 / Q - 1}
    \bullet \underbracket[.4pt]{\circ \circ \cdots \circ}_{1 / Q - 1}
    \cdots \bullet \underbracket[.4pt]{\circ \circ \cdots \circ}_{1 / Q -
    1}.
  \end{equation}
  Any attempt to move a fermion would make the largest space bigger than
  $1 / Q - 1$.
\end{proof}

Thus, any state will have a space that is larger than or equal to $1 / Q - 1$.
Due to the system's periodicity, we can therefore fix the first $1 / Q$ sites to be
\begin{equation}
  \bullet \underbracket[.4pt]{\circ \circ \cdots \circ}_{1 / Q - 1}.
  \label{fixingsites}
\end{equation}
This leaves us with a smaller subspace of the full basis to generate: 
the system with size $(N - 1) / Q$ and $N - 1$ particles.

\begin{theorem}
  \label{property2}For any ground state of the system, the largest space must
  not exceed $p$ sites.
\end{theorem}

\begin{proof}
  Assume that there exists a ground state unit cell with the largest space equal to $p + 1$
  sites. We can write it as
  \begin{equation}
    \underbracket[.4pt]{\bullet ? ? \cdots ? \bullet}_{\tmop{Block} A} 
    \underbracket[.4pt]{\circ \circ \cdots \circ}_p \circ .\label{untrueGS}
  \end{equation}
  Let $E_A$ be the energy of the block $A$, so that the
  energy density of this ground state is $\frac{E_A}{N / Q}$. Let us construct
  the following unit cell, which consists of $p$ consecutive ground-state unit cells
  ($\ref{untrueGS}$):
  \begin{equation}
		\mbox{\footnotesize ${\medmuskip=1mu 
    \underbracket[.4pt]{
		\begin{array}{|l|}
    \hline
		\underbracket[.4pt]{\bullet ? ? \cdots ? \bullet}_{\tmop{Block} A} 
    \underbracket[.4pt]{\circ \circ \cdots \circ}_p \circ \\
		\hline
		\end{array}\begin{array}{|l|}
    \hline
    \underbracket[.4pt]{\bullet ? ? \cdots ? \bullet}_{\tmop{Block} A} 
    \underbracket[.4pt]{\circ \circ \cdots \circ}_p \circ  \\
		\hline
		\end{array}
		\cdots
		\begin{array}{|l|}
    \hline
		\underbracket[.4pt]{\bullet ? ? \cdots ? \bullet}_{\tmop{Block} A} 
    \underbracket[.4pt]{\circ \circ \cdots \circ}_p \circ \\
		\hline
		\end{array}
		}_p
		}$}.
  \end{equation}
  This unit cell has the same energy density $\frac{p E_A}{p N / Q} =
  \frac{E_A}{N / Q}$ as the ground-state unit cell (\ref{untrueGS}). Let us now move the
  additional empty spaces to the end of this chain, which still does not change
  the energy density:
  \begin{equation}
		\mbox{\footnotesize ${\medmuskip=1mu 
		\begin{array}{|l|}
    \hline
    \underbracket[.4pt]{\bullet ? ? \cdots ? \bullet}_{\tmop{Block} A} 
    \underbracket[.4pt]{\circ \circ \cdots \circ}_p \\
		\hline
		\end{array}\begin{array}{|l|}
    \hline
    \underbracket[.4pt]{\bullet ? ? \cdots ? \bullet}_{\tmop{Block} A} 
    \underbracket[.4pt]{\circ \circ \cdots \circ}_p \\
		\hline
		\end{array}
		\cdots
		\begin{array}{|l|}
    \hline
    \underbracket[.4pt]{\bullet ? ? \cdots ? \bullet}_{\tmop{Block} A} 
    \underbracket[.4pt]{\circ \circ \cdots \circ}_p  \underbracket[.4pt]{\circ
    \circ \cdots \circ}_p  \\
		\hline
		\end{array}
		}$}.
		\label{untrue2}
  \end{equation}
  Now, let us assume that the last fermion in block $A$ contributes to the
  potential energy of this block by amount $E_{\Delta}$. If $E_{\Delta}=0$, we can
  always swap this last fermion with the rest of block $A$ and again consider the 
	last fermion of a new block. If we take this last fermion out and
	replace it with a hole, then the energy of the block $A$ will
  decrease by $E_{\Delta}$. Let us now move the last fermion in unit cell
  (\ref{untrue2}) by $p$ sites to the right:
  \begin{equation}
		\mbox{\footnotesize ${\medmuskip=1mu 
		\begin{array}{|l|}
    \hline
    \underbracket[.4pt]{\bullet ? ? \cdots ? \bullet}_{\tmop{Block} A} 
    \underbracket[.4pt]{\circ \circ \cdots \circ}_p \\
		\hline
		\end{array}\begin{array}{|l|}
    \hline
    \underbracket[.4pt]{\bullet ? ? \cdots ? \bullet}_{\tmop{Block} A} 
    \underbracket[.4pt]{\circ \circ \cdots \circ}_p \\
		\hline
		\end{array}
		\cdots
    \begin{array}{|l|}
    \hline
		\underbracket[.4pt]{\bullet ? ? \cdots ? \circ}_{\tmop{Block} A'} 
    \underbracket[.4pt]{\circ \cdots \circ}_{p - 1} \bullet
    \underbracket[.4pt]{\circ \circ \cdots \circ}_p \\
		\hline
		\end{array}
		}$},
  \end{equation}
  where block $A'$ is block $A$ with the last fermion replaced by a hole. Block $A'$
  has energy $E_A - E_{\Delta}$. The last fermion does not contribute now to the
  overall potential energy, because it is surrounded by $p$ sites on both
  sides. Such a unit cell now has energy density
  \begin{equation}
    \frac{p E_A - E_{\Delta}}{p N / Q} = \frac{E_A}{N / Q} -
    \frac{E_{\Delta}}{p N / Q},
  \end{equation}
  which is lower than the energy of the ground-state unit cell $\left( \ref{untrueGS}
  \right)$, and this leads to a contradiction. A similar process can be used to
  show that a ground state cannot have a space equal to $p + 2$ and more
  sites. Thus, we conclude that the largest space in any ground state must have 
	at most $p$ sites.
\end{proof}

Using Theorems \ref{property1} and \ref{property2}, we can significantly
decrease the number of generated states.

\subsection{Details of the calculations}

Our calculations were performed using Mathematica.{\cite{Mathematica}} Firstly a partial basis
for a specific number of particles $N$, density $Q$ and interaction range $p$ was generated.
States of this partial basis had the first $1 / Q$ sites fixed to the configuration shown in
Eq.\ (\ref{fixingsites}) by Theorem \ref{property1}, and any states that were
not in agreement with Theorem \ref{property2} were removed. 
Then, the energy density was
calculated\footnote{Instead of using loop statements to fill in tables, it was found
that the Mathematica's \tmtexttt{Table[]} function was more
efficient and the time consumption
scaled exactly exponentially with the system size.} for every state and this list
of energies was simplified by removing duplicates.
In order to discard the energies that cannot describe the ground state, the expression $\forall_\beta E_\alpha<E_\beta$ was
assessed\footnote{Mathematica's \tmtexttt{Reduce[]} function
was found to give the most reliable simplification results, however it also
needed much higher computational resources. \tmtexttt{Simplify[]} and
\tmtexttt{FullSimplify[]} were found to be quite similar in resource
consumption and the latter was chosen due to higher reliability for
simplifying complicated conditions.}. Some energies however could not be compared without
knowing the values of $\{ U_m \}$. The final list contains
the energies of all phases that have the lowest energy for some set values of $\{ U_m
\}$; these are the CDW phases of the system.

\subsection{Results for $Q = 1 / (p - 1)$}

Unit cells and energy densities for $p=3,4$, and $5$ are presented in Table\ \ref{tab:results1}. Due to the finite size of the systems studied, we can only look for CDW unit cells up to a specific size ($L_{\tmop{max}}$). Phase diagrams in Figure\ \ref{fig:phases} show what phases are expected to appear for different values of the potentials $\{ U_m \}$.

\begin{table}
\caption{Ground-state (GS) unit cells and their energies in a system with $Q = 1 / (p - 1)$. $f$ is the degeneracy of the ground state. $L_{\tmop{max}}$ is the maximum size of the unit cell that was analyzed. Colors designate phases shown in Fig.\ \ref{fig:phases} (color online). \label{tab:results1}}
\begin{center}
\begin{tabular}{lcccc}
	\hline \hline
	System & GS unit cell & Energy density & $f$ & \\
	\hline
	$p=3,$              & {\medmuskip=0mu $\bullet \circ$} & $\frac{1}{2} U_2$ & 2 & {\color[rgb]{0,0.666,1}$\blacksquare$}\\
	$Q=1/2,$            & {\medmuskip=0mu $\bullet \bullet \circ \circ$} & $\frac{1}{4} (U_1 + U_3)$ & 4 & {\color[rgb]{0.333,1,0}$\blacksquare$}\\
	$L_{\tmop{max}}=28$ & {\medmuskip=0mu $\bullet \bullet \bullet \circ \circ \circ$} & $\frac{1}{6} (2 U_1 + U_2)$ & 6 & {\color[rgb]{1,1,0}$\blacksquare$}\\
	\hline
	$p=4,$              & {\medmuskip=0mu $\bullet \circ \circ$} & $\frac{1}{3} U_3$ & $3$ & {\color[rgb]{0,0.666,1}$\blacksquare$}\\
	$Q=1/3,$            & {\medmuskip=0mu $\bullet \bullet \circ \circ \circ \circ$} & $\frac{1}{6} U_1$ & $6$ & {\color[rgb]{0.333,1,0}$\blacksquare$}\\
	$L_{\tmop{max}}=36$ & {\medmuskip=0mu $\bullet \circ \bullet \circ \circ \circ$} & $\frac{1}{6} (U_2 + U_4)$ & $6$ & {\color[rgb]{1,1,0}$\blacksquare$}\\
	& {\medmuskip=0mu $\bullet \bullet \circ \circ \circ \bullet \circ \circ \circ$} & $\frac{1}{9} (U_1 + 2 U_4)$ & $9$ & {\color[rgb]{1,0,0}$\blacksquare$}\\
	& {\medmuskip=0mu $\bullet \circ \bullet \circ \bullet \circ \circ \circ \circ$} & $\frac{1}{9} (2 U_2 + U_4)$ & $9$ & {\color[rgb]{0.5,0,0.5}$\blacksquare$}\\
	& {\medmuskip=0mu $\bullet \circ \bullet \circ \circ \bullet \circ \bullet \circ \circ \circ \circ$} & $\frac{1}{12} (2 U_2 + U_3)$ & $12$ & {\color[rgb]{1,0.5,0}$\blacksquare$}\\
	& {\medmuskip=0mu $\bullet \bullet \bullet \circ \circ \circ \circ \bullet \circ \bullet \circ \circ \circ \circ \bullet \circ \bullet \circ \circ \circ \circ$} & $\frac{1}{21} (2 U_1 + 3 U_2)$ & $21$ & {\color[rgb]{0,0,0}$\blacksquare$}\\
	\hline
	$p=5,$              & {\medmuskip=0mu $\bullet \circ \circ \circ$} & $\frac{1}{4} U_4$ & 4 \\
	$Q=1/4,$            & {\medmuskip=0mu $\bullet \circ \circ \bullet \circ \circ \circ \circ$} & $\frac{1}{8} (U_3 + U_5)$ & 8 \\
	$L_{\tmop{max}}=32$ & {\medmuskip=0mu $\bullet \circ \bullet \circ \circ \circ \circ \circ$} & $\frac{1}{8} U_2$ & 8 \\
	& {\medmuskip=0mu $\bullet \circ \bullet \circ \circ \circ \circ \bullet \circ \circ \circ \circ$} & $\frac{1}{12} (U_2 + 2 U_5)$ & 12 \\
	& {\medmuskip=0mu $\bullet \circ \circ \bullet \circ \circ \bullet \circ \circ \circ \circ \circ$} & $\frac{1}{12} 2 U_3$ & 12 \\
	& {\medmuskip=0mu $\bullet \bullet \circ \circ \circ \circ \bullet \circ \circ \circ \circ \bullet \circ \circ \circ \circ$} & $\frac{1}{16} (U_1 + 3 U_5)$ & 16 \\
	& {\medmuskip=0mu $\bullet \bullet \circ \circ \circ \circ \circ \bullet \bullet \circ \circ \circ \circ \circ \bullet \circ \circ \circ \circ \circ$} & $\frac{1}{20} 2 U_1$ & 20 \\
	\hline \hline
\end{tabular}
\end{center}
\end{table}

\begin{figure}
\begin{center}
\includegraphics[width=\columnwidth]{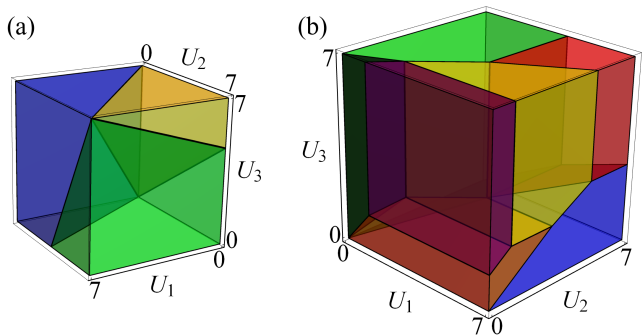}
\caption{(Color online) Phase diagrams of (a) $p = 3, Q = 1 / 2$ and (b) $p = 4, Q = 1 / 3$ with $U_4=1$. For color legend see Table\ \ref{tab:results1}. \label{fig:phases}}
\end{center}
\end{figure}

\subsection{Results for $Q = 1 / (p - 2)$}

Table\ \ref{tab:results2} presents the unit cells and energy densities for $p=4$ and $5$. Notice that for $p=5$, we have found ground-state unit cells up to $L_{\tmop{max}}$ and thus potentially there could be a ground state containing an even larger unit cell that was not found in our calculation.

\begin{table}
\caption{As Table\ \ref{tab:results1}, but for a system with density $Q = 1 / (p - 2)$.\label{tab:results2}}
\begin{center}
		\begin{tabular}{ccc}
			\hline \hline
			GS unit cell & Energy density & $f$\\
			\hline
			\multicolumn{3}{c}{$p = 4, Q = 1/2, L_{\tmop{max}}=26$} \\
			\hline
			{\medmuskip=0mu $\bullet \circ$} & $\frac{1}{2} (U_2 + U_4)$ & 2 \\
			{\medmuskip=0mu $\bullet \bullet \circ \circ$} & $\frac{1}{4} (U_1 + U_3 + 2 U_4)$ & 4 \\
			{\medmuskip=0mu $\bullet \bullet \bullet \circ \circ \circ$} & $\frac{1}{6} (2 U_1 + U_2 +
			U_4)$ & 6 \\
			{\medmuskip=0mu $\bullet \bullet \bullet \bullet \circ \circ \circ \circ$} & $\frac{1}{8}
			(3 U_1 + 2 U_2 + U_3)$ & 8 \\
			{\medmuskip=0mu $\bullet \bullet \circ \bullet \circ \circ \bullet \circ$} & $\frac{1}{8}
			(U_1 + 2 U_2 + 3 U_3)$ & 8 \\
			\hline
			\multicolumn{3}{c}{$p = 5, Q = 1/3, L_{\tmop{max}}=21$} \\
			\hline
			{\medmuskip=0mu $\bullet \circ \circ$} & $\frac{1}{3} U_3$ & $3$\\
			{\medmuskip=0mu $\bullet \circ \bullet \circ \circ \circ$} & $\frac{1}{6} (U_2 + U_4)$ &
			$6$\\
			{\medmuskip=0mu $\bullet \bullet \circ \circ \circ \circ$} & $\frac{1}{6} (U_1 + U_5)$ &
			$6$\\
			{\medmuskip=0mu $\bullet \bullet \circ \bullet \circ \circ \circ \circ \circ$} &
			$\frac{1}{9} (U_1 + U_2 + U_3)$ & $2 \times 9$\\
			{\medmuskip=0mu $\bullet \bullet \circ \circ \circ \bullet \circ \circ \circ$} &
			$\frac{1}{9} (U_1 + 2 U_4 + 2 U_5)$ & $9$\\
			{\medmuskip=0mu $\bullet \circ \bullet \circ \bullet \circ \circ \circ \circ$} &
			$\frac{1}{9} (2 U_2 + U_4 + U_5)$ & $9$\\
			{\medmuskip=0mu $\bullet \circ \bullet \circ \circ \bullet \circ \bullet \circ \circ \circ
			\circ$} & $\frac{1}{12} (2 U_2 + U_3 + 3 U_5)$ & $12$\\
			{\medmuskip=0mu $\bullet \bullet \circ \circ \bullet \circ \circ \circ \bullet \circ \circ
			\circ \bullet \circ \circ$} & $\frac{1}{15} (U_1 + 2 U_3 + 4 U_4)$ & $15$\\
			{\medmuskip=0mu $\bullet \bullet \bullet \circ \circ \circ \circ \circ \bullet \circ
			\bullet \circ \circ \circ \circ$} & $\frac{1}{15} (2 U_1 + 2 U_2 + U_5)$ &
			$2 \times 15$\\
			{\medmuskip=0mu $\bullet \bullet \bullet \circ \circ \circ \circ \circ \bullet \bullet
			\circ \circ \circ \circ \circ$} & $\frac{1}{15} (3 U_1 + U_2)$ & 15\\
			{\medmuskip=0mu $\bullet \bullet \circ \circ \bullet \bullet \circ \circ \circ \circ \circ
			\bullet \bullet \circ \circ \circ \circ \circ$} & $\frac{1}{18} (3 U_1 +
			U_3 + 2 U_4 + U_5)$ & $18$\\
			{\medmuskip=0mu $\bullet \bullet \bullet \circ \circ \circ \circ \circ \bullet \circ
			\bullet \circ \bullet \circ \circ \circ \circ \circ$} & $\frac{1}{18} (2
			U_1 + 3 U_2 + U_4)$ & $18$\\
			{\medmuskip=0mu $\bullet \bullet \circ \circ \bullet \circ \circ \bullet \bullet \circ
			\circ \circ \circ \circ \bullet \bullet \circ \circ \circ \circ \circ$} &
			$\frac{1}{21} (3 U_1 + 2 U_3 + 2 U_4)$ & $21$\\
			{\medmuskip=0mu $\bullet \bullet \bullet \circ \circ \circ \circ \bullet \circ \bullet
			\circ \circ \circ \circ \bullet \circ \bullet \circ \circ \circ \circ$} &
			$\frac{1}{21} (2 U_1 + 3 U_2 + 3 U_5)$ & $21$\\
			\hline \hline
		\end{tabular}
\end{center}
\end{table}

\subsection{Discussion of the results}

Our results illustrate how highly nontrivial and unpredictable the ground-state configurations are for critical densities higher than $Q=1/p$. For example, for a half-filled system ($Q=1/2$), judging only from the $p=3$ case, one would naively expect a similar trend to be present in all other cases: for all units cells to consist of a chain of occupied sites, followed by a chain of the same length, but with empty sites. However, Table\ \ref{tab:results2} shows that for $p=4$ there exists a ground state with a unit cell {\medmuskip=0mu $(\bullet\nobreak\bullet\nobreak\circ\nobreak\bullet\nobreak\circ\nobreak\circ\nobreak\bullet\nobreak\circ)$}, which does not follow this prediction. Therefore, it is very difficult to create a simple set of rules describing the ground-state properties of all the phases in the system with high critical density.

We also conclude that the number of possible CDW phases in the system grows with the maximum interaction range $p$ and the density $Q$. For example, in the system $p=5, Q=1/3$ presented in Table\ \ref{tab:results2}, there are at least 14 different CDW phases, and we expect $p = 6, Q = 1 / 4$ to contain even more.

For $t\neq 0$, we expect non-CDW phases to be present in the system. If one considers the phase diagrams from Fig.\ \ref{fig:phases}, on the interfaces between any two phases there probably are Luttinger liquid and bond-order phases, similarly to the findings of Refs.\ {\onlinecite{Schmitteckert2004}} and {\onlinecite{Mishra2011}}. Therefore, if our assumption that the number of phases grows quickly with the maximum interaction range is correct, then we can predict that for high $p$, the phase diagram consists of mainly non-CDW phases, while CDW insulators are only present when certain $U_m$ are very high. Thus, a large interaction range may imply the loss of insulating properties of the material.

\section{Summary and outlook} \label{summary}

We have studied the ground-state properties of the extended $t$-$V$ model on a lattice with a potential that does not 
necessarily satisfy Eq.\ (\ref{condU}). We have shown how to 
construct the ground state of all the Mott insulating phases at low critical 
densities, and we have calculated the ground-state unit cells of a few example cases for 
higher critical densities. Thus, we provide a description of possible 
CDW phases of the system with any interaction range and any critical 
density in the atomic limit.

One could also work beyond the atomic limit $(t>0)$, in which case
other, non-CDW phases would be present in the system. However, to find a 
simple description of all the phases while varying the interaction range 
could prove to be more difficult. 
To simulate such systems, we propose to use 
matrix product states {\cite{PerezGarcia2007,Verstraete2008}}, 
due to their recent achievements in calculations of 
lattice models using relatively low resources.

\begin{acknowledgments}
M.S. is fully funded by EPSRC, NoWNano DTC grant number EP/G03737X/1.
The author would like to thank Evgeni Burovski for his supervision
on the initial work with the generalized $t$-$V$ model.
The author would also like to express his gratitude to Neil Drummond for useful discussions and proofreading.
\end{acknowledgments}

\bibliography{LatticePRB}

\end{document}